\documentclass[twocolumn]{autart}
\usepackage{graphicx} 

\usepackage{amstext}
\usepackage{amsmath}
\usepackage{amssymb}
\usepackage{amsfonts}
\usepackage{booktabs}
\usepackage{subfigure}
\usepackage{xcolor}
\usepackage{color}
\usepackage{multirow}
\newtheorem{theorem}{Theorem}
\newtheorem{assumption}{Assumption}
\newtheorem{property}{Property}
\usepackage{soul}
\usepackage{array}
\usepackage{appendix}
\usepackage{makecell}
\usepackage{wrapfig}
\usepackage{tabu}
\usepackage{url}
\newtheorem{proof}{Proof}
\newtheorem{corollary}{Corollary}
\newtheorem{remark}{Remark}
\begin{document}
\begin{frontmatter}
\title{MAKO: Meta-Adaptive Koopman Operators for Learning-based Model Predictive Control of Parametrically Uncertain Nonlinear Systems}
\thanks{{This research is supported by the National Research Foundation, Singapore, and PUB, Singapore’s
National Water Agency under its RIE2025 Urban Solutions and Sustainability
(USS) (Water) Centre of Excellence (CoE) Programme, awarded to Nanyang
Environment \& Water Research Institute (NEWRI), Nanyang Technological
University, Singapore (NTU). This research is also supported by the Ministry of Education, Singapore, under its Academic Research Fund Tier 1 (RG63/22). Any opinions, findings and conclusions or recommendations expressed in this material are those of the author(s) and do not reflect the views of the National Research Foundation, Singapore, and PUB, Singapore’s National Water Agency.}}
\thanks{Corresponding author: X. Yin. Email: xunyuan.yin@ntu.edu.sg.}
\author[newri,cceb]{Minghao Han},  
\author[eth]{Kiwan Wong},  
\author[newri,nus]{Adrian Wing-Keung Law},
\author[newri,cceb]{Xunyuan Yin}

\address[newri]{Nanyang Environment and Water Research Institute (NEWRI), Nanyang Technological University, Singapore}
\address[cceb]{School of Chemistry, Chemical Engineering and Biotechnology, Nanyang Technological University, Singapore}
\address[eth]{Soft Robotics Lab, ETH Zurich, Switzerland}
\address[nus]{Department of Civil and Environmental Engineering, National University of Singapore, Singapore}




\begin{abstract}
In this work, we propose a meta-learning-based Koopman modeling and predictive control approach for nonlinear systems with parametric uncertainties. An adaptive deep meta-learning-based modeling approach, called Meta Adaptive Koopman Operator (MAKO), is proposed. Without knowledge of the parametric uncertainty, the proposed MAKO approach can learn a meta-model from a multi-modal dataset and efficiently adapt to new systems with previously unseen parameter settings by using online data. Based on the learned meta Koopman model, a predictive control scheme is developed, and the stability of the closed-loop system is ensured even in the presence of previously unseen parameter settings. Through extensive simulations, our proposed approach demonstrates superior performance in both modeling accuracy and control efficacy as compared to competitive baselines.
\end{abstract}


\end{frontmatter}
\section{Introduction}



Parametric uncertainties are common in nonlinear systems, often arising from factors such as variations in payload and operating conditions \cite{hong2006adaptive,tao2014multivariable}. The presence of these uncertainties can cause performance degradation and instability and pose great challenges to the design of control systems. The endeavor to ensure desired control system tracking performance and system stability has driven the advancement of adaptive control methods for nonlinear systems with parametric uncertainties over recent decades.


Model predictive control (MPC) is a popular advanced control technique \cite{mayne2000constrained}. It optimizes the future predicted behavior of the system by utilizing a dynamic model along with current measurement data.
Adaptive model predictive control (AMPC) has been proposed to address uncertainties \cite{fukushima2007adaptive,zhang2020adaptive}; however, the results on nonlinear systems have been limited. In \cite{mayne1993adaptive}, a receding horizon predictive control scheme was developed for nonlinear systems that are subject to control constraints and are linear in the unknown parameters. In \cite{adetola2009adaptive}, robust MPC was integrated with a min-max approach and a Lipschitz-based approach for online parameter update. In \cite{zhang2022self}, {a} self-triggered AMPC method was proposed for constrained discrete-time nonlinear systems facing parametric uncertainties and additive disturbances. More results on nonlinear AMPC can be found in \cite{dehaan2007adaptive,xu2023multi}. Although these results represent significant advancements, first-principle models are typically required as the foundation of control system designs. The theoretical assumption of linear dependency on uncertain parameters further limits its applicability to general nonlinear processes.

Recently, the Koopman operator theory has gained substantial research attention, owing to its capability to represent the dynamics of complex nonlinear processes in a linear manner \cite{koopman1931hamiltonian}. 
Several algorithms have been developed to construct linear Koopman models from process data. These include dynamic mode decomposition (DMD) \cite{schmid2010dynamic}  and extended dynamic mode decomposition (EDMD) \cite{li2017extended}. {DMD represents the observables directly in the original state space, while EDMD employs a predetermined set of basis functions; both methods solve least-squares problems to approximate the linear Koopman operator.} {\cite{zeng2024sampling} proposed a sampling theorem for the exact identification of continuous-time nonlinear dynamical systems using the Koopman operators.}
To streamline the design of the observable functions, researchers have proposed various ML-enabled Koopman modeling and control methods {\cite{morton2018deep,yeung_learning_2019,han2020deep,azencot2020forecasting,shi2022deep,deka2022koopman,han2023robust_tii,li2024machine}}. {In \cite{hao2024deep}, the authors have provided convergence guarantees for the error as the capacity of neural network (NN)  increases, and have derived the error upper bound with connection to the spectral property of the Koopman operator. }
However, these existing approaches have primarily targeted addressing a specific control task with fixed model parameters.

In this work, we aim to exploit the Koopman operator framework and machine learning (ML) to facilitate the modeling and control of parametrically uncertain nonlinear systems. Within the context of ML-based modeling, we acknowledge that uncertainties may be introduced in a beneficial manner to mitigate overfitting (e.g. Monte-Carlo dropout) as well as to aid the assessment of uncertainties of the ML predictions \cite{wei2023probabilistic}. Nonetheless, 
the learning-based modeling and control of parametrically uncertain systems can be conceptualized as a multi-task problem, which can be effectively addressed using the meta-learning concept \cite{hospedales2021meta}. {Meta-learning, often referred to as ``learning to learn'', is a paradigm in ML that focuses on developing algorithms that are capable of generalizing across tasks. Meta-learning aims to extract and utilize knowledge from multiple related tasks, enabling rapid adaptation to new, unseen tasks with minimal data or computational effort.}
{Recently, significant advancements have been achieved in integrating meta-learning with control.
In \cite{molybog2021does}, the optimization landscape of the model-agnostic meta-learning (MAML) algorithm was investigated, with a focus on identifying conditions that guarantee its global convergence in a single task LQR setting. 
\cite{musavi2023convergence} established sufficient conditions for the stability of the dynamical system during optimization and proved that MAML converges to a stationary point in the multi-task LQR setting. 
\cite{toso2024meta} introduced a MAML-based method for solving LQR problems in multi-task, heterogeneous settings, and has provided personalization guarantees for both model-based and model-free learning.}
{In addition}, meta-reinforcement learning (meta RL), a fusion of meta-learning and RL, has been developed for learning-based control in multi-task scenarios. Meta RL controllers utilize previously acquired knowledge and real-time data to adapt to new tasks, wherein the system dynamics, objectives, or distribution of noise and disturbance can vary.
\cite{nagabandi2018learning} introduced a meta-learning-based MPC framework capable of fine-tuning the meta-trained NN model using online data. This method was applied to control a legged robot in the presence of changed payloads, terrains, and even a disabled leg. 
In \cite{mcclement2022meta}, a novel offline meta RL strategy was proposed for tuning proportional-integral controllers in process control systems. 
However, the online adaptation of deep NNs is inefficient and computationally demanding. Furthermore, it is generally challenging for these meta-RL approaches to offer stability and ensure closed-loop performance. 

\begin{figure}
    \centering    \includegraphics[width=0.8\columnwidth]{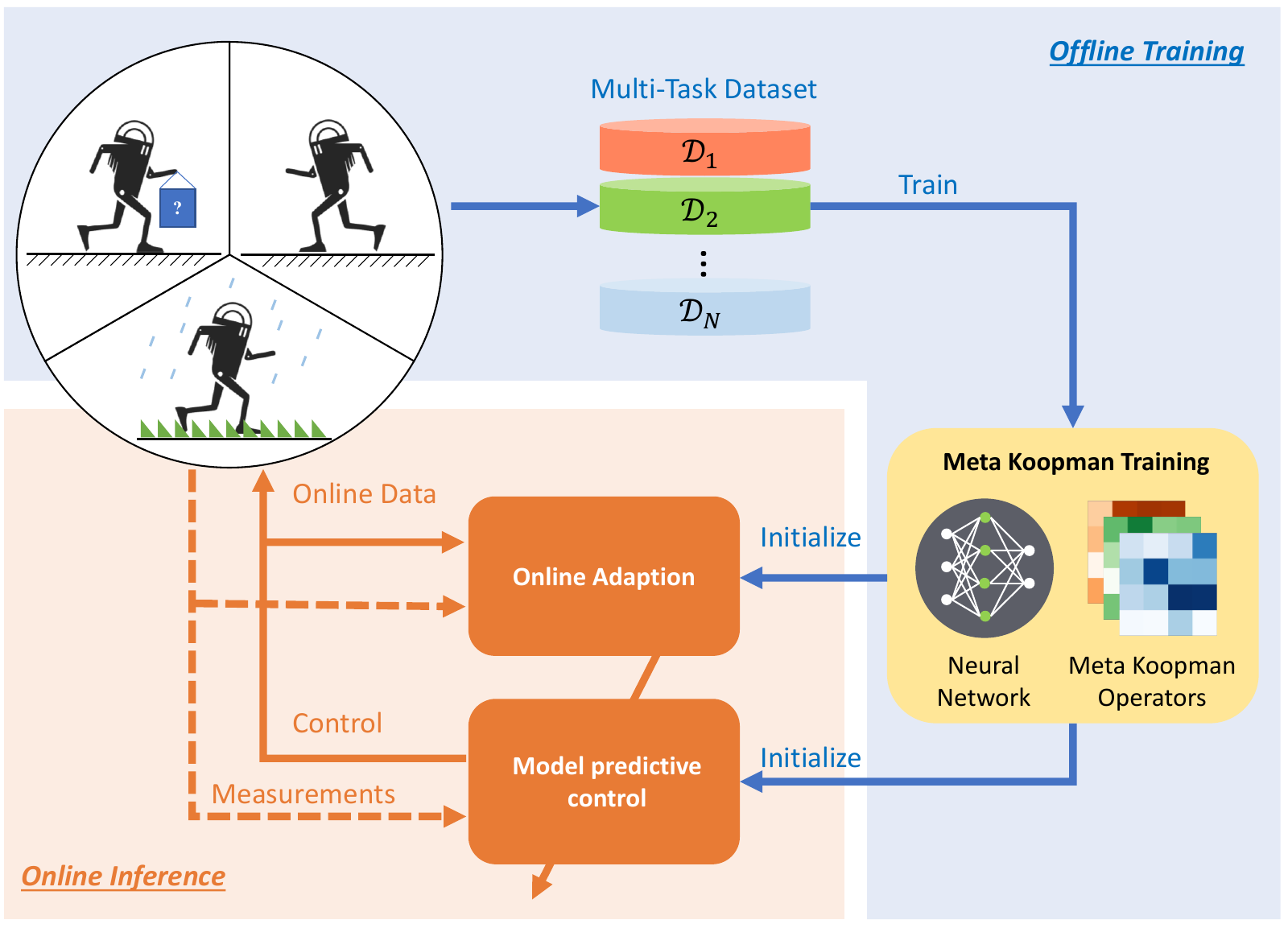}
    \caption{An overview of the Meta Koopman pipeline. 
    }
    \label{fig:overview}
\end{figure}

Based on these observations, we aim to integrate meta-learning with Koopman operator theory to create a learning-based adaptive control framework for parametrically uncertain nonlinear systems. Within the Koopman operator framework, we proposed a meta-adaptive Koopman operator (MAKO) modeling approach. This approach learns from a multi-modal dataset to construct a meta-model for online adaptation. An adaptation scheme is developed to update the meta-model using online data while ensuring convergence. Based on the adaptive meta-model, a predictive control scheme is proposed for the underlying uncertain nonlinear systems. The contributions of this work include: 1) Meta-learning and Koopman operator theory are integrated for the first time to establish a learning-based adaptive MPC framework applicable to a general class of parametrically uncertain nonlinear systems. 2) We rigorously prove the convergence of both the model online adaptation and the closed-loop system. 3) Based on three benchmark systems from various fields, MAKO demonstrates good modeling accuracy and robust tracking control performance in the presence of parameter uncertainties, and it outperforms competitive baselines.

\section{Preliminaries}

\subsection{Meta-learning}

Meta-learning is concerned with developing automatic learning algorithms that can leverage data from previous tasks to quickly adapt to new tasks with trials. The new tasks may differ from the previous tasks in terms of system dynamics, noise distributions, and control objectives \cite{hospedales2021meta}.  In this work, we focus on scenarios where the parameters of the nonlinear system vary in different task settings, that is,
\begin{equation}
    x_{k+1} = f(x_k, u_k, \Theta), \text{ s.t. } \Theta\sim p(\Theta)
    \label{eq: meta nonlinear system}
\end{equation}
where $f$ is an unknown nonlinear function of the state $x_k\in\mathcal{X}\subset\mathbb{R}^n$, the control input $u_{k}\in\mathcal{U}\subset\mathbb{R}^m$, and the system parameter $\Theta {\in\Xi\subset\mathbb{R}^l}$. We use $\mathcal{X}$ and $\mathcal{U}$ to denote the state and input sets, respectively. {We use $\Xi$ to denote the space of parameters.} The parameter of the system, denoted by $\Theta$, follows an unknown distribution $p(\Theta)$, and each instance of $\Theta$ corresponds to a specific task setting. {In the following, we denote a sampled instance of $\Theta$ as $\Theta^i$.}

We formulate the supervised meta-model learning problem as 
\begin{gather}
    \min_{\theta} \mathbb{E}_{p(\Theta^i)} \left [ \mathcal{L}\left(\mathcal{D}_{i}, \theta\right) \right]\\
    \mathcal{L}(\mathcal{D}_{i}, \theta) = \frac{1}{T}\sum_{k=1}^T\left\Vert x^i_{k+1} - {\hat{f}}_\theta\left(x^i_k, u^i_k\right)\right\Vert
\end{gather}
\vspace{-2.1em}

\noindent where ${\hat{f}}_\theta$ is a parameterized model to approximate the unknown dynamic function $f$ in (\ref{eq: meta nonlinear system}), and $\theta$ denotes the parameters to be optimized. We use $\mathcal{D}_i:=\{[x^i_k,u^i_k]^T_{k=1}|\Theta^i\}$ to denote {the sub-dataset of $T$ steps of state-input data  collected under a task setting $\Theta^i$}. The meta-dataset $\mathcal{D}_{\Theta}:=\{\mathcal{D}_i\}_{i=1}^N$ comprises $N$ sub-datasets. In this work, we present a meta-Koopman framework to learn a model that can effectively adapt to new tasks.

\subsection{The Koopman operator}\label{subsec:Koopman}
In this subsection, we briefly introduce the ideas and notations of the Koopman operator theory. The Koopman theory was first formulated in \cite{koopman1931hamiltonian}. 
According to the Koopman theory, a general nonlinear system of the form $x_{k+1} = f(x_k), k\in \mathbb{N}$, can be transformed into a linear system within an infinite-dimensional function space $\mathcal{G}$. 
This space encompasses all square-integrable real-valued functions defined over the compact domain $\mathcal{X}$. The elements of $\mathcal{G}$, denoted as ${\phi}$, are referred to as \textit{observables}. The Koopman operator $\mathcal{K}:\mathcal{G} \rightarrow \mathcal{G}$ satisfies the relation 
$
{\phi}\circ f(x_{k}) = \mathcal{K}{\phi}(x_{k})
$, 
where $\circ$ denotes function composition, and ${\phi}\in \mathcal{G}$ represents the observable function. 
While initially proposed for autonomous nonlinear systems, the concept of the Koopman operator has been extended to controlled systems in recent years \cite{korda2018linear,proctor2018generalizing,iacob2024koopman}.
For controlled systems, the Koopman operator adheres to the following condition: ${\phi}_x\circ f(x_k, u_k) = A {\phi}_x(x_k) + B {\phi}_u(x_k, u_k)
$, 
where $A$ and $B$ are submatrices of the Koopman operator, and ${\phi}_x$ and ${\phi}_u$ represent the observables for the state $x_k$ and control input $u_k$, respectively.
In practical applications, {it is often relevant to find a finite-dimensional numerical approximation of $\mathcal{K}$ within a finite-dimensional function space $\overline{\mathcal{G}}\subset \mathcal{G}$ for the controller design \cite{korda2018linear,zhang2022robust}}. This space is defined by a set of linearly independent observables {$\{\psi|\psi:\mathbb{R}^{n}\rightarrow\mathbb{R}^h\}$}. {To facilitate controller design and analysis, the Koopman operator is typically approximated in a space where the control inputs act linearly \cite{korda2018linear,zhang2022robust}, i.e. ${\phi}_u(x_k, u_k) = u_k$. }

\section{Method}
In this section, we elaborate on the architecture of the proposed meta-adaptive Koopman operator (MAKO) model learning approach and elucidate how it learns and adapts to new tasks. 

\subsection{Meta-trained Koopman model}

\begin{figure}
    \centering
    \includegraphics[width=0.5\textwidth]{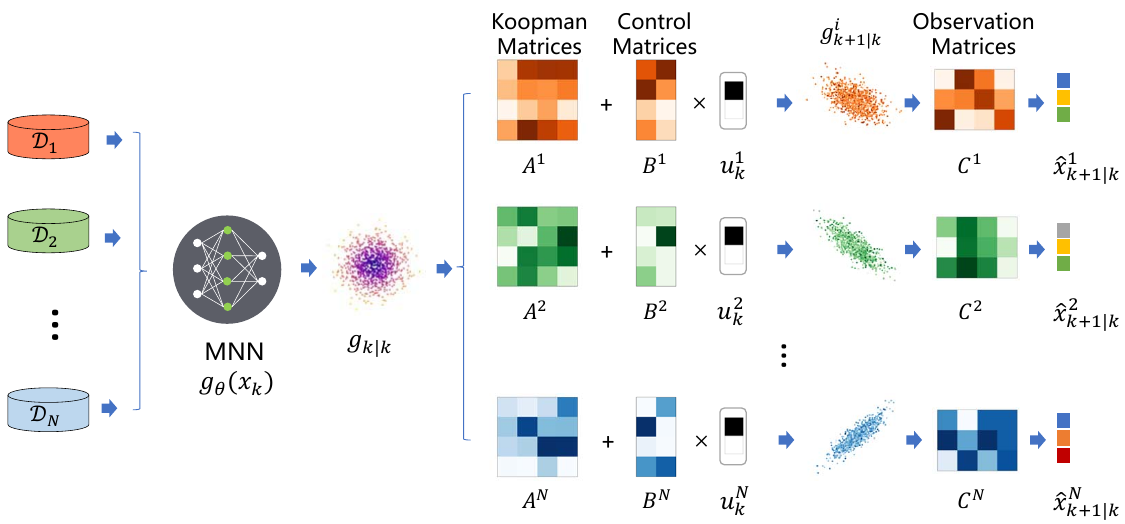}
    \caption{An overview of the proposed MAKO pipeline. The meta-trained neural network encodes a common set of observables. On the observable space, individual Koopman operators are trained for each $\mathcal{D}_i$ to predict future observables. The observables are transformed back to the state space with $C^i$ to predict the future states. }
    \label{fig:pipeline}
\end{figure}

The MAKO model comprises two trainable building blocks: a meta-trained neural network (MNN) responsible for parameterizing the observable functions, and a set of linear Koopman operators for predicting future states and observables across various tasks. 
An overview of our pipeline involving the MNN and the Koopman operators is presented in Fig.~\ref{fig:pipeline}. 

\subsubsection{Meta-trained neural network}

In the proposed framework, the {MNN} plays the key role of encoding an informative observable space, which is shared across different task settings. The key insight lies in recognizing that while the dynamics of the system vary across different tasks, the latent variables that characterize the system dynamics shall remain consistent. Inspired by this understanding, the observables of the system under different task settings are encoded as follows,
\begin{equation}
    g^i_k = {\psi}_\theta\left(x^i_k\right), x^i_k \in \mathcal{D}_i
\end{equation}
where ${\psi}_\theta(\cdot)$ is a multi-layer neural network parameterized by the trainable parameters $\theta$. 
As shown in Fig.~\ref{fig:pipeline}, the sub-datasets collected under different task settings are passed through the same MNN. Consequently, the MNN may not distinguish the task settings and focus on extracting the key dynamic features shared across different settings. 

The encoding mechanism of the MNN, which maps the original state space to a latent observable space, is conceptually similar to the encoder in Koopman autoencoders \cite{azencot2020forecasting,deka2022koopman}. Nonetheless, the MNN in our approach focuses on creating a shared representation across multi-modal task settings to support adaptation rather than reconstructing the states in the nominal setting, which is different from previous autoencoders in \cite{azencot2020forecasting,deka2022koopman}.

\subsubsection{Koopman operator}
In the encoded observable space, the Koopman operators propagate the observables forward to predict future states. 
The dynamic behavior of the system varies across different task settings. Accordingly, for each task setting $\Theta^i$, one set of Koopman operators $A^i$, $B^i$, and $C^i$ will be learned to characterize the specific dynamic behavior.
For a dataset consisting of $N$ sub-datasets, $N$ sets of Koopman operators will be learned. The dynamic behavior of the Koopman system under $\Theta^i$ is described as:
\begin{subequations}\label{Koopman operators}
\begin{align}
     g^i_{k+1|k}&=A^ig^i_{k|k}+B^iu^i_k\\
     {\hat{x}^i}_{k+1|k}&=C^ig^i_{k+1|k}
\end{align}
\end{subequations}
The Koopman operators ${A^i,B^i,C^i}$ in (\ref{Koopman operators}) and the parameters of the MNN, denoted by $\theta$, are trainable.  

\subsection{Meta-learning}
In this subsection, we elaborate on the details of learning a meta-adaptive Koopman model that can be effectively adapted to different task settings. 
The goal of meta-model learning is to find a set of parameters $\theta, {A^i,B^i,C^i}$ that minimize the multi-step-ahead prediction error characterized by:
\vspace{-0.3em}
\begin{subequations}\label{eq:loss}
\begin{gather}
    \mathcal{L}\left(\theta, {A^i,B^i,C^i}\right) = \mathbb{E}_{{p(\Theta^i)}}\frac{1}{HT}\sum_{k=1}^{T}\sum_{t=1}^{H} \Vert x_{k+t}^{{i}}-{C^i}g_{k+t|k}^{{i}}\Vert_2^2 \label{eq: loss a} \\
    g_{k+t|k}^{{i}} = {A^i} g_{k+t-1|k}^{{i}} + {B^i} u_{k+t-1}^{{i}} \\ 
    g_{k|k}^{{i}}= {\psi}_\theta\left(x^{{i}}_k\right), \{x_k^{{i}}, u_k^{{i}}\} \in \mathcal{D}_{{i}}
\end{gather}
\end{subequations}
where $H$ denotes the prediction horizon {and $T$ denotes the length of data trajectories}. 
The MAKO model is trained to minimize the expected prediction error under the unknown task distribution $p\left(\Theta^{{i}}\right)$, which in practice can be approximated by using a meta-dataset $\mathcal{D}_{\Theta}$ consisting of multiple sub-datasets $\mathcal{D}_i$, $i=1,\ldots,N$. 

Incorporating the concepts above, the optimization problem for MAKO modeling is formulated as follows:
\begin{subequations}\label{eq:loss_final}
\begin{align}
        \min_{\theta, \{A^i,B^i,C^i\}} \frac{1}{NTH}&\sum_{i=1}^N\sum_{k=1}^{T}\sum_{t=1}^{H}\Vert x^i_{k+t}-C^i g^{i}_{k+t|k}\Vert_2^2 \\
        \text{s.t. }g^{i}_{k+t|k} &= A^ig^{i}_{k+t-1|k}+B^iu^i_{k+t-1}\label{eq:seed propogation}\\ 
        g^{i}_{k|k}& = {\psi}_\theta\left(x^i_k\right), \{x_k^{{i}}, u_k^{{i}}\}\in \mathcal{D}_{{i}}\label{eq:seed observable}
\end{align}
\end{subequations}
{
\begin{remark}
    The proposed scheme is applicable to systems beyond linear time-invariant systems considered in \cite{molybog2021does,musavi2023convergence,toso2024meta}. However, we acknowledge the challenge of establishing a tight bound for the generalization and performance of the meta-trained Koopman model without access to the explicit forms of $p(\Theta)$ and $f(x, u, \Theta)$. This limitation contrasts with prior work, such as \cite{molybog2021does,musavi2023convergence,toso2024meta}, where such bounds have been more rigorously analyzed. Establishing performance guarantees for meta-learning in various types of parametrically uncertain systems may be a promising topic to explore in the future. 
\end{remark}}
\subsection{Online adaptation}

We elaborate on how to adapt the meta-trained Koopman model, which is presented in the previous subsection, to new tasks using online data. 

First, the Koopman operators learned on the meta-dataset can be combined to serve as an initial approximation of the exact Koopman operators in the new setting, denoted as ${\hat{A}_0,\hat{B}_0,\hat{C}_0}=\left\{\frac{1}{N}\sum A^i,\frac{1}{N}\sum B^i,\frac{1}{N}\sum C^i\right\}$.
Let ${\hat{\Psi}}_k := [\hat{A}_k, \hat{B}_k]$ denote the approximated Koopman operators at instant $k$, ${\hat{\Psi}}_0 =[\hat{A}_0, \hat{B}_0]$. $X_k:=[g_k^\mathrm{T},u_k^\mathrm{T}]^\mathrm{T}$ denotes the extended observable, one has the observable prediction error given by ${\widetilde{g}}_{k+1}:=g_{k+1} -\hat{\Psi}_{{k}}\ X_k$. In addition, the prediction error of the state is denoted as  ${\widetilde{x}}_{k+1}:=x_{k+1} -\hat{C}_{{k}} g_{k+1}$. Define the cost function for the state and observable prediction error,
$J(\hat{\Psi}_{{k}},\hat{C}_{{k}}) := \Vert g_{k+1}-{\hat{\mathrm{\Psi}}_{{k}}} X_k\Vert_2^2 + \Vert x_{k+1}-\hat{C}_{{k}}g_{k+1}\Vert_2^2$. 
To further refine the model, the online data generated during the exploitation stage will be used through gradient descent.
The gradient of $J$ with respect to $\hat{\Psi}_{{k}}$ and $\hat{C}_{{k}}$ can be obtained following 
\begin{subequations}
    \begin{align}
    \nabla_{\hat{\Psi}} J_k &:= \frac{\partial J\left(\hat{\Psi}_{{k}},\hat{C}_{{k}}\right)}{\partial\hat{\Psi}_{{k}}} = -X_k \widetilde{g}_{k+1}^\mathrm{T} \\
    \nabla_{\hat{C}} J_k &:= \frac{\partial J\left(\hat{\Psi}_{{k}},\hat{C}_{{k}}\right)}{\partial\hat{C}_{{k}}} = -g_{k+1} \widetilde{x}_{k+1}^\mathrm{T}
\end{align}
\label{eq: parameter gradient}
\end{subequations}
The update law for the Koopman operators is given as 
\begin{subequations}\label{eq:update law}
    \begin{align}
         \hat{\Psi}_{k+1}&=\hat{\Psi}_{k} - \lambda_k \nabla_{\hat{\Psi}} J_k^\mathrm{T}={\hat{\Psi}}_k+\lambda_k{\widetilde{g}}_{k+1}X_k^\mathrm{T}\\
         \hat{C}_{k+1}&=\hat{C}_{k} - \lambda_k \nabla_{\hat{C}} J_k^\mathrm{T}=\hat{C}_k+\lambda_k\widetilde{x}_{k+1}g_{k+1}^\mathrm{T}
    \end{align}
\end{subequations}
where $\lambda_k$ denotes the learning rate at time instant $k$. In the online adaptation, we propose to use an adaptive learning rate adopted from \cite{zhu2015adaptive} in the following form:
\begin{equation}
    \lambda_ k = \min\left(\frac{2-\alpha}{X^\mathrm{T}_k X^\mathrm{}_k}, \frac{2-\alpha}{g_{k{{+1}}}^\mathrm{T} g^\mathrm{}_{k{+1}}}\right)\label{eq:learning rate}
\end{equation}
where $\alpha$ is a pre-determined hyperparameter subject to $0<\alpha<2$.

\subsubsection{{Nominal adaptation}}
{We first consider the case where the exact Koopman operators exist on the finite-dimensional observable space. }
Before establishing the theoretical guarantee, we introduce the following assumptions. 
{
\begin{assumption}\label{assumption: bounded}
    The sets $\mathcal{X}$ and $\mathcal{U}$ are compact, and the set $\mathcal{X}$ is forward invariant
for the system $f$; that is, for any $\Theta\in\Xi$, $x\in \mathcal{X}$ and $u\in \mathcal{U}$, the
inclusion $f(x,u,\Theta)\in \mathcal{X}$ holds.
\end{assumption}
}

\begin{assumption}\label{assumption: 2}
    The lifting function ${\psi}_\theta(\cdot)$ and the system dynamics $f(\cdot, \cdot, \Theta)$ are continuous on $\mathcal{X}$ {and $\text{ }\mathcal{U}$}, $\forall \Theta \in \Xi$.
\end{assumption}

\begin{assumption}\label{assumption: 1}
    For each task setting $\Theta$, there exists a set of Koopman operators $A,B,C$ on the observable space encoded by the ${\psi}_\theta(\cdot)$, such that $f(x, u, \Theta) = C(A {\psi}_\theta(x) + Bu)$,
    for all $x\in{\mathcal{X}}$ and $u\in{\mathcal{U}}$.
\end{assumption}
Assumption~\ref{assumption: 1} essentially assumes the existence of {exact} Koopman operators. Based on the above assumptions, we present the following theorem on the convergence of the parameter approximation and model prediction errors. Let $\widetilde{\Psi}_k: = \Psi-\hat{\Psi}_k$ and $\widetilde{C}_k: = C- \hat{C}_k$ denote the parameter approximation error, where $\Psi:=\left[A, B\right]$ and $C$ denote the {exact Koopman operators}. 
Before proceeding, we introduce relevant properties of the trace of matrices to facilitate the proof.
\begin{property}\label{property}
For matrices ${M_1}$ and ${M_2}$, and vectors $x$ and $y$, all with proper dimensions, the following properties hold:
\begin{itemize}
    \item $\mathrm{tr}\left({M_1M_2}\right) = \mathrm{tr}\left({M_2M_1}\right)$
    \item $\mathrm{tr}\left({M_1+ M_2}\right) = \mathrm{tr}\left({M_1}\right) + \mathrm{tr}\left({M_2}\right)$
    \item $\mathrm{tr}\left(yx^\mathrm{T}\right) = \mathrm{tr}\left(x^\mathrm{T}y\right)$
\end{itemize}
where $\mathrm{tr}(\cdot)$ denotes the trace of a given matrix. 
\end{property}

\begin{theorem}\label{theorem: adaptation stability}
Consider the uncertain nonlinear system \eqref{eq: meta nonlinear system} with uncertain parameters $\Theta$ and the corresponding Koopman operators $A$, $B$, and $C$.  If Assumptions~\ref{assumption: bounded}-\ref{assumption: 1} hold, with the adaptive updating laws \eqref{eq:update law} and \eqref{eq:learning rate}, the parameter approximation errors $\widetilde{\Psi}_{{k}}$ and $\widetilde{C}_{{k}}$ are ultimately bounded and the predicted
state error $\widetilde{x}$ asymptotically converges to zero.
\end{theorem}

\vspace{-1.7em}
\begin{proof}
The proof of Theorem~\ref{theorem: adaptation stability} is based on that of Theorem 1 in \cite{zhu2015adaptive}. First, let us select a Lyapunov candidate $V_k := \mathrm{tr}(\widetilde{\Psi}_k^{\mathrm{T}}\widetilde{\Psi}_k) + \mathrm{tr}(\widetilde{C}_k^{\mathrm{T}}\widetilde{C}_k)$. It can be derived that 
\begin{equation}\label{eq:proof:123}
\begin{aligned}
    V_{k+1} = & V_k + \mathrm{tr}\left(2\lambda_k\widetilde{\Psi}_k^\mathrm{T} \nabla_{\hat{\Psi}} J_k^\mathrm{T} + \lambda_k^2 \nabla_{\hat{\Psi}} J_k\nabla_{\hat{\Psi}} J_k^\mathrm{T}\right)\\
    &+ \mathrm{tr}\left(2\lambda_k\widetilde{C}_k^\mathrm{T} \nabla_{\hat{C}} J_k^\mathrm{T} + \lambda_k^2 \nabla_{\hat{C}} J_k\nabla_{\hat{C}} J_k^\mathrm{T}\right)
\end{aligned}
\end{equation}
By substituting \eqref{eq: parameter gradient} into (\ref{eq:proof:123}), the second and the third terms on the right-hand-side of (\ref{eq:proof:123}) can be computed as 
$\mathrm{tr}(2\lambda_k\widetilde{\Psi}_k^\mathrm{T} \nabla_{\hat{\Psi}} J^\mathrm{T}_k + \lambda_k^2 \nabla_{\hat{\Psi}} J_k\nabla_{\hat{\Psi}} J^\mathrm{T}_k)
=\lambda_k\mathrm{tr}(-2\widetilde{\Psi}_k^\mathrm{T}\widetilde{g}_{k+1}X_k^\mathrm{T} + \lambda_k X_k\widetilde{g}^\mathrm{T}_{k+1}\widetilde{g}_{k+1}X_k^\mathrm{T})
=\\\lambda_k(-2X_k^\mathrm{T}\widetilde{\Psi}_k^\mathrm{T}\widetilde{g}_{k+1} + \lambda_k X_k^\mathrm{T}X_k\widetilde{g}^\mathrm{T}_{k+1}\widetilde{g}_{k+1})
=\lambda_k(-2\widetilde{g}_{k+1}^\mathrm{T}\widetilde{g}_{k+1}\\ + \lambda_k X_k^\mathrm{T}X_k\widetilde{g}^\mathrm{T}_{k+1}\widetilde{g}_{k+1})
= \lambda_k(-2 + \lambda_k X_k^\mathrm{T}X_k)\widetilde{g}^\mathrm{T}_{k+1}\widetilde{g}_{k+1}$.
By taking \eqref{eq:learning rate} into account, one has
$\mathrm{tr}(2\lambda_k\widetilde{\Psi}_{{k}}^\mathrm{T} \nabla_{\hat{\Psi}} J_k^\mathrm{T} + \lambda_k^2 \nabla_{\hat{\Psi}} J_k\nabla_{\hat{\Psi}} J_k^\mathrm{T})\leq -\lambda_k \alpha \widetilde{g}^\mathrm{T}_{k+1}\widetilde{g}_{k+1}$.
Following similar derivations as above, the following inequality holds, 
$\mathrm{tr}(2\lambda_k\widetilde{C}_{{k}}^\mathrm{T} \nabla_{\hat{C}} J_k^\mathrm{T} + \lambda_k^2 \nabla_{\hat{C}} J_k\nabla_{\hat{C}} J_k^\mathrm{T})\leq -\lambda_k \alpha \widetilde{x}^\mathrm{T}_{k+1}\widetilde{x}_{k+1}$.
Therefore, it follows from \eqref{eq:proof:123} that
\begin{equation}
    V_{k+1} \leq V_k  -\lambda_k \alpha \left( \widetilde{g}^\mathrm{T}_{k+1}\widetilde{g}_{k+1} + \widetilde{x}^\mathrm{T}_{k+1}\widetilde{x}_{k+1}\right)\label{eq:proof-4}
\end{equation}
which implies that $V_k$ is decreasing as $k$ increases.  Furthermore, since $V_k\geq 0$, $\lim_{k\xrightarrow{}+\infty}V_k$ exists, and the parameter estimation errors $\widetilde{\Psi}$ and $\widetilde{C}$ are ultimately bounded. 

Applying \eqref{eq:proof-4} to all the time instants and aggregating the resulting inequalities yields:
\begin{equation}
    V_{k+1}\leq V_1 -  \alpha \sum_{i=1}^k \lambda_{{i}}\left(\widetilde{g}^\mathrm{T}_{{i}+1}\widetilde{g}_{{i}+1} + \widetilde{x}^\mathrm{T}_{{i}+1}\widetilde{x}_{{i}+1}\right)
\end{equation}
Consequently, 
$\alpha\sum_{k=1}^{+\infty} \lambda_k\left(\widetilde{g}^\mathrm{T}_{k+1}\widetilde{g}_{k+1} + \widetilde{x}^\mathrm{T}_{k+1}\widetilde{x}_{k+1}\right) \leq V_1 -  \lim_{k\xrightarrow{}+\infty}V_{k+1}$.  {Since $\mathcal{X}$ and $\mathcal{U}$ are bounded by Assumption~\ref{assumption: bounded}, and the lifting function $\psi_\theta(\cdot)$ is continuous according to Assumption~\ref{assumption: 2}, it can be inferred that $X^\mathrm{T}_k X_k$ and $g^\mathrm{T}_{k+1} g_{k+1}$ are upper bounded for all $k$. Consequently, there exists a positive constant $\underline{\lambda}$ such that $\lambda_k\geq\underline{\lambda}$ for all $k$. Therefore, we have $\alpha\underline{\lambda}\sum_{k=1}^{+\infty}\left(\widetilde{g}^\mathrm{T}_{k+1}\widetilde{g}_{k+1} + \widetilde{x}^\mathrm{T}_{k+1}\widetilde{x}_{k+1}\right) \leq V_1$.}
This indicates the convergence of the infinite series on the left-hand side, which implies $\widetilde{g} \rightarrow 0$ and $\widetilde{x} \rightarrow 0$.
\end{proof}

\subsubsection{{Robust adaptation}}
{
Assumption~\ref{assumption: 1} requires the existence of a set of Koopman operators that exactly characterize the original system on a finite-dimensional observable space. However, for the general class of nonlinear systems, particularly those with parametric uncertainties, the existence of a finite-dimensional invariant subspace cannot be guaranteed \cite{zeng2024sampling}. In the following, we consider the more practical case where the exact Koopman operators do not exist and the modeling errors are present, as follows:
\begin{subequations}\label{eq: robust Koopman system}
\begin{align}
    {\psi}_\theta(f(x_k, u_k, \Theta)) &= A {\psi}_\theta(x_k) + Bu_k +w_k\\
    x_k& = C {\psi}_\theta(x_k) +v_k
\end{align}
\end{subequations}

\vspace{-1.5em}

Following the derivation in \cite{zhang2022robust}, we can infer that the modeling errors $w\in \mathcal{W}\subset \mathbb{R}^h$ and $v\in \mathcal{V}\subset \mathbb{R}^n$ are bounded {based on Assumptions~\ref{assumption: bounded} and \ref{assumption: 2}}. Therefore, there exist positive constants $\epsilon_w$ and $\epsilon_v$ such that $\Vert w \Vert \leq \epsilon_w, \forall w \in \mathcal{W}$ and $\Vert v \Vert \leq \epsilon_v, \forall v \in \mathcal{V}$.

Next, we propose a robust adaptation scheme. The objective function containing the modeling errors is 
$\bar{J}(\hat{\Psi}_k,\hat{C}_k, w_k, v_k) = \Vert g_{k+1}-{\hat{\mathrm{\Psi}}_k} X_k -w_k\Vert_2^2 + \Vert x_{k+1}-\hat{C}_kg_{k+1} -v_k\Vert_2^2$.
To propose a robust adaptation scheme, we first introduce the ideal noise as 
\begin{align}\label{eq: optimal noise}
    w_k^*, v_k^* &= \min_{\substack{w_k\in\mathcal{W} \\v_k\in\mathcal{V}}} \bar{J}\left(\hat{\Psi}_k,\hat{C}_k, w_k, v_k\right)
\end{align}
The gradient of $\bar{J}$ with respect to $\hat{\Psi}_k$ and $\hat{C}_k$ can be obtained following 
$\nabla_{\hat{\Psi}} \bar{J}_k = -X_k (\widetilde{g}_{k+1}-w^*_k)^\mathrm{T}$, $
\nabla_{\hat{C}} \bar{J}_k = -g_{k+1} (\widetilde{x}_{k+1}-v^*_k)^\mathrm{T}$.
The robust update law at each time step is given as 
\begin{subequations}\label{eq: robust update law}
    \begin{align}
         \hat{\Psi}_{k+1}&={\hat{\Psi}}_k+\lambda_k(\widetilde{g}_{k+1}-w^*_k)X_k^\mathrm{T}\\\hat{C}_{k+1}&=\hat{C}_k+\lambda_k(\widetilde{x}_{k+1}-v^*_k)g_{k+1}^\mathrm{T}
    \end{align}
\end{subequations}
where $\lambda_k$ is determined following \eqref{eq:learning rate}.

\begin{theorem}\label{theorem: robust adaptation stability}
Consider the uncertain nonlinear system \eqref{eq: meta nonlinear system} with uncertain parameters $\Theta$. If {Assumptions~\ref{assumption: bounded} and \ref{assumption: 2}} hold, then under the updating laws \eqref{eq:learning rate}, \eqref{eq: optimal noise} and \eqref{eq: robust update law}, the parameter approximation errors $\widetilde{\Psi}_k$, $\widetilde{C}_k$ are ultimately bounded, and $\lim_{k\rightarrow\infty}\Vert\widetilde{x}_{k}\Vert \leq \epsilon_v$. 
\end{theorem}
\vspace{-0.5em}
\begin{proof}
The proof of Theorem~\ref{theorem: robust adaptation stability} follows a similar procedure as in Theorem~\ref{theorem: adaptation stability}. We adopt $V_k := \mathrm{tr}(\widetilde{\Psi}_k^{\mathrm{T}}\widetilde{\Psi}_k) + \mathrm{tr}(\widetilde{C}_k^{\mathrm{T}}\widetilde{C}_k)$ as the Lyapunov function. Let $\Delta(\widetilde{\Psi}_k,\nabla_{\hat{\Psi}} \bar{J}_k):=\mathrm{tr}(2\lambda_k\widetilde{\Psi}_k^\mathrm{T} \nabla_{\hat{\Psi}} \bar{J}_k^\mathrm{T} + \lambda_k^2 \nabla_{\hat{\Psi}} \bar{J}_k\nabla_{\hat{\Psi}} \bar{J}_k^\mathrm{T})$. With the update law described in \eqref{eq: robust update law}, 
$\Delta(\widetilde{\Psi}_k,\nabla_{\hat{\Psi}} \bar{J}_k)
=\lambda_k\mathrm{tr}(-2\widetilde{\Psi}_k^\mathrm{T}(\widetilde{g}_{k+1} -w_k^*)X_k^\mathrm{T}+ \lambda_k X_k(\widetilde{g}_{k+1}-w_k^*)^\mathrm{T}(\widetilde{g}_{k+1}-w_k^*)X_k^\mathrm{T})
= \lambda_k(-2X_k^\mathrm{T}\widetilde{\Psi}_k^\mathrm{T}(\widetilde{g}_{k+1}-w_k^*) 
+ \lambda_k X_k^\mathrm{T}X_k\Vert\widetilde{g}_{k+1}-w_k^*\Vert_2^2)
= \lambda_k(-2(\widetilde{g}_{k+1}-w_k)^\mathrm{T}(\widetilde{g}_{k+1}-w_k^*) 
+ \lambda_k X_k^\mathrm{T}X_k\Vert\widetilde{g}_{k+1}-w_k^*\Vert_2^2)$. 
Consider that $w_k^*$ is the optimal solution to \eqref{eq: optimal noise}, it can be inferred that 
$(\widetilde{g}_{k+1}-w_k)^\mathrm{T}(\widetilde{g}_{k+1}-w_k^*) \geq \Vert\widetilde{g}_{k+1}-w_k^*\Vert_2^2$. 
It follows that $\Delta(\widetilde{\Psi}_k,\nabla_{\hat{\Psi}} \bar{J}_k)\leq \lambda_k\left(-2 + \lambda_k X_k^\mathrm{T}X_k\right)\Vert\widetilde{g}_{k+1}-w_k^*\Vert_2^2$. 
Incorporating the learning rate \eqref{eq:learning rate}, one has 
$\Delta(\widetilde{\Psi}_k,\nabla_{\hat{\Psi}} \bar{J}_k)\leq -\alpha\lambda_k\Vert\widetilde{g}_{k+1}-w_k^*\Vert_2^2$.
Apply similar derivation to $\Delta(\widetilde{C}_k,\nabla_{\hat{C}} \bar{J}_k)$, and combine with \eqref{eq:proof:123}, one has 
\begin{equation}
    V_{k+1} \leq V_k  -\alpha\lambda_k  \left( \Vert\widetilde{g}_{k+1}-w_k^*\Vert_2^2 + \Vert\widetilde{x}_{k+1}-v_k^*\Vert_2^2\right)\label{eq:proof2-4}
\end{equation}
Therefore, it is proved that the parameter approximation errors $\widetilde{\Psi}_k$ and $\widetilde{C}_k$ are ultimately bounded. Aggregating the inequalities \eqref{eq:proof2-4} for all time instants, it can be inferred that $\lim_{k\rightarrow\infty}\Vert\widetilde{g}_{k+1} - w^*_k\Vert=0$ and $\lim_{k\rightarrow\infty}\Vert\widetilde{x}_{k+1} - v^*_k\Vert=0$ {based on Assumptions~\ref{assumption: bounded} and \ref{assumption: 2}}. In addition, according to Assumption~\ref{assumption: 2}, $\mathcal{W}$ and $\mathcal{V}$ are bounded, and $\Vert w^*_k\Vert \leq \epsilon_w$ and $\Vert v^*_k\Vert \leq \epsilon_v$. Therefore, we have $\lim_{k\rightarrow\infty}\Vert\widetilde{x}_{k+1}\Vert \leq \epsilon_v$, which concludes the proof.
\end{proof}

\vspace{-0.5em}
\begin{remark}
    The continuity of ${\psi}_\theta(\cdot)$, as considered in Assumption 2, can be guaranteed by adopting activation functions that are continuous, such as ReLU, ELU, and Sigmoid etc. 
\end{remark}

\begin{remark}
The proofs of Theorem~\ref{theorem: adaptation stability} and Theorem~\ref{theorem: robust adaptation stability} are built based on the theoretical results in \cite{zhu2015adaptive}. Compared to \cite{zhu2015adaptive}, the theoretical contributions of this current work are two-fold. 
First, \cite{zhu2015adaptive} relies on the assumption of a small, constant learning rate to ensure convergence, while this work employs a dynamic learning rate as defined in \eqref{eq:learning rate}. This dynamic learning rate adapts based on the state and input data at each time step, and this eases the restrictive condition of requiring a fixed, small learning rate.
Second, while \cite{zhu2015adaptive} focuses on nominal linear systems, this work extends the Koopman-based framework to nonlinear systems with parametric uncertainties. Modeling residuals on the finite-dimensional space are further considered, which is different from the noise-free setting in \cite{zhu2015adaptive}. 
\end{remark}
}
\vspace{-0.3em}
\section{Meta-Koopman-based Adaptive Model Predictive Control}
In this section, we propose an adaptive model predictive control (AMPC) approach based on the learned MAKO model. First, we present the MAKO-based AMPC design. Subsequently, we establish the stability criterion for the resulting closed-loop system. 

\vspace{-0.8em}

\subsection{MAKO-based Adaptive MPC}
Based on the meta Koopman model, we design an MPC scheme \cite{borrelli2017predictive} to solve the finite horizon optimal control problem, minimizing the cumulative stage cost. In the nominal setting, the dynamics of ${\hat g}_k$ can be described by the MAKO model as follows:
\begin{subequations} \label{eq:nominal dynamic}
\begin{align}
    \hat{g}_{k+1|k} =& \hat{A}_k \hat{g}_{k|k} + \hat{B}_k u_{k}\\
    \hat{x}_{k+1|k} = &\hat{C}_k \hat{g}_{k+1|k}
\end{align}
\end{subequations}
where $\hat{g}_{k|k}$ is the nominal observable encoded by the MNN. 
The AMPC solves the following deterministic optimal tracking control problem 
\begin{subequations}\label{eq:nominal MPC}
    \begin{align}
    \min_{u_{k:k+T}}V_k = & \sum_{t=1}^{T} \Vert \hat{C}_k\hat{g}_{k+t|k}-x_s\Vert^2_Q + \Vert \Delta u_{k+t|k}\Vert_R^2\label{eq:nominal MPC1}\\\label{eq:nominal MPC2}
    \text{s.t. } &\hat{g}_{k+t+1} = \hat{A}_k \hat{g}_{k+t|k} + \hat{B}_k u_{k+t|k}\\
    & \hat{C}_k\hat{g}_{k+T+1|k} = x_s\label{eq:nominal MPC3}
    \\&\Delta u_{k+t|k} = u_{k+t|k}- u_{k+t-1|k}
    \\&u_{k+t|k}\in \mathcal{U}
\end{align}
\end{subequations}
where $x_s$ is the given set-point in the original state space, $Q$ and $R$ are known positive definite weighting matrices, and $\Vert \cdot \Vert^2_W$ represents the weighted Euclidean norm with $W$ being a positive definite matrix. At each sampling instant, the AMPC problem \eqref{eq:nominal MPC} is solved and the corresponding optimal control input $u^*_{k|k}$ is applied to the system in (\ref{eq: meta nonlinear system}). In the meantime, the Koopman matrices $\hat{A}_k, \hat{B}_k$ and $ \hat{C}_k$ are updated according to the update law described in \eqref{eq:update law}.

\subsection{Stability with learned Koopman operators}
Next, we prove the stability of the closed-loop system based on the AMPC controller in \eqref{eq:nominal MPC}. 


\begin{theorem}\label{prop:stability proposition}
Consider the nonlinear system \eqref{eq: meta nonlinear system} with the adaptive updating law in \eqref{eq:update law} and \eqref{eq:learning rate}, and the MPC controller in \eqref{eq:nominal MPC}. If {Assumptions~\ref{assumption: bounded}-\ref{assumption: 1} hold,} then the tracking error of the closed-loop system is asymptotically stable. 
\end{theorem}
\begin{proof}
By solving the MPC problem \eqref{eq:nominal MPC}, we can obtain the feasible optimal control sequence at time $k$, 
$\{u^*_{k|k},u^*_{k+1|k}, \dots, u^*_{k+T|k} \}$,
and the resulting predicted optimal observable trajectory $\{\hat{g}^*_{k|k},\hat{g}^*_{k+1|k}, \dots, \hat{g}^*_{k+T|k}, \\\hat{g}^*_{k+T+1|k} \}$.
Consider an extended input trajectory 
\begin{equation}\label{eq:subopt control seq}
\{u^*_{k+1|k}, \dots, u^*_{k+T|k}, u^*_{k+T+1|k} \},
\end{equation}
where the control input at time instant $k+T+1$ remains unchanged from $u^*_{k+T|k}$, i.e., $u^*_{k+T+1|k}=u^*_{k+T|k}$. 

Due to the update of the Koopman operators, there exists a deviation $\delta$ between the state trajectory $\{\hat{x}^*_{k+t|k} \}$ predicted at instant $k$ and $\{\hat{x}^*_{k+t|k+1} \}$ predicted at $k+1$, where $t$ refers to the time step within the prediction horizon. $\delta_k^t := (\hat{C}_{k+1}\hat{A}^t_{k+1} - \hat{C}_{k}\hat{A}^t_{k})\hat{g}^*_{k+1|k} + \sum_{j=1}^t(\hat{C}_{k+1}\hat{A}^{j-1}_{k+1}\hat{B}_{k+1} - \hat{C}_{k}\hat{A}^{j-1}_{k}\hat{B}_{k})u^*_{k+j|k}$. 
{Denote the update step as $\Delta \hat{\Psi}_k:=\hat{\Psi}_{k+1}-\hat{\Psi}_{k}$ and $\Delta \hat{C}_k:=\hat{C}_{k+1}-\hat{C}_{k}$. It can be observed that $\delta_k^t$ goes to zero, if $\Delta \hat{\Psi}_k$ and $\Delta \hat{C}_k$ approach zero.}

In the following, we adopt the value function $V^*_k = \sum_{t=1}^{T} \Vert \hat{C}_k\hat{g}^*_{k+t|k} -x_s\Vert^2_Q + \Vert \Delta u^*_{k+t|k}\Vert_R^2$ as the Lyapunov candidate. The suitability of $V_k$ as a Lyapunov candidate has been proved in \cite{grune2017nonlinear}.
Consider the Lyapunov candidate $V^*_{k+1}$ at time $k+1$, $V^*_{k+1} = \sum_{t=2}^{T+1} \left(\Vert \hat{C}_{k+1}\hat{g}^*_{k+t|k+1} - x_s\Vert^2_Q + \Vert \Delta u^*_{k+t|k+1}\Vert_R^2\right)$. 
Due to optimality, the value of $V^*_{k+1}$ is no greater than the value function of the sub-optimal solution, that is,
$V^*_{k+1}
\leq \sum_{t=2}^{T+1}( \Vert \hat{C}_{k}\hat{g}^*_{k+t|k} -x_s + \delta^t_k\Vert^2_Q + \Vert \Delta u^*_{k+t|k} \Vert_R^2)
\leq \sum_{t=2}^{T+1} (\Vert \hat{C}_k\hat{g}^*_{k+t|k}-x_s\Vert^2_Q+\Vert \delta^t_k\Vert^2_Q + \Vert \Delta u^*_{k+t|k}\Vert_R^2)  
= V^*_k + \Vert \hat{C}_k\hat{g}^*_{k+T+1|k}-x_s\Vert^2_Q + \Vert \Delta u^*_{k+T+1|k}\Vert_R^2- \Vert \hat{C}_k\hat{g}^*_{k+1|k}-x_s\Vert^2_Q - \Vert \Delta u^*_{k+1|k}\Vert_R^2+ \sum_{t=2}^{T+1} \Vert \delta^t_k\Vert^2_Q$.
It follows from \eqref{eq:nominal MPC3} and \eqref{eq:subopt control seq} that $\hat{C}_k\hat{g}^*_{k+T+1|k} = x_s$ and $\Delta u^*_{k+T+1|k} = \mathbf{0}$. Note that $\lim_{k\rightarrow\infty}\Vert\widetilde{g}_k\Vert=0$ and $\lim_{k\rightarrow\infty}\Vert\widetilde{x}_k\Vert=0$ according to Theorem~\ref{theorem: adaptation stability}, hence $\Delta\hat{\Psi}_k$ and $\Delta\hat{C}_k$ go to zero, and the residual approaches zero $\lim_{k\rightarrow\infty}\delta_k^t = 0$ as a consequence. Therefore, it can be inferred that the tracking error of the nominal system in \eqref{eq:nominal dynamic} is asymptotically stable. Furthermore, incorporate the result from Theorem~\ref{theorem: adaptation stability} that the state prediction error {asymptotically converges to zero}, the closed-loop system is asymptotically stable.  
\end{proof}

{
\begin{corollary}\label{corollary:1}
    Consider the nonlinear system \eqref{eq: meta nonlinear system} with the adaptive updating law given by \eqref{eq:learning rate}, \eqref{eq: optimal noise}, and \eqref{eq: robust update law}, and the MPC controller in \eqref{eq:nominal MPC}. If {Assumptions ~\ref{assumption: bounded} and \ref{assumption: 2} hold}, then the tracking error of the closed-loop system is ultimately bounded. 
\end{corollary}

\begin{proof}
    The proof of Corollary~\ref{corollary:1} follows the same procedure adopted in the proof of Theorem~\ref{prop:stability proposition}, which establishes the asymptotic stability of the tracking error of the nominal system \eqref{eq:nominal dynamic}. By further incorporating that $\lim_{k\rightarrow\infty}\Vert\widetilde{x}_{k+1}\Vert \leq \epsilon_v$ according to Theorem~\ref{theorem: robust adaptation stability}, the tracking error of the closed-loop system can be proved to be ultimately bounded. 
\end{proof}
}
\section{Results}

In this section, the proposed MAKO learning-based control framework is evaluated and compared to a competitive baseline on three benchmark examples via simulations. {The codes for reproducing our results can be found at the link provided in the footnote\footnote{\url{https://github.com/hithmh/Meta-Koopman}}.}

\subsection{Simulation setup and baselines}\label{subsec:Simulation setup and baselines}


\subsubsection{Cartpole}
\vspace{-0.3em}
First, we consider a classic cartpole balancing problem {\cite{sutton2018reinforcement}}. The controller is expected to maintain the pendulum in its upright, vertical orientation. The state vector comprises $[x, \dot x, \theta, \dot \theta]^\text{T}$, where $x$ denotes the horizontal position of the cart and $\theta$ denotes the angular position of the pole in rads. The action is the horizontal force applied to the cart ($u\in[-20, 20]$). $x_{\text{threshold}}$ and $\theta_{\text{threshold}}$ represents the maximum position and angle, respectively, with $x_{\text{threshold}}=10$ and $\theta_{\text{threshold}}=20^\circ$. An episode terminates prematurely if $\vert \theta\vert> \theta_{\text{threshold}}$ . The stage cost for control performance evaluation is given by $c = 0.1*(x/x_{\text{threshold}})^2 + (\theta/\text{threshold})^2$. In the Cartpole system, the length of the pole $l_p$ and the mass of the pole $m_p$ are considered uncertain, with $l_p\in [0.1m, 1.0m]$ and $m_p\in [0.01 kg, 0.2kg]$. Their nominal values are $l_p = 0.5m $ and $m_p = 0.1kg$. The value of $\alpha$ is set to 1.995, and the weighting matrices are $Q=\text{diag}([0.01, 0, 1, 0.2])$, where $\text{diag}(\cdot)$ denotes constructing a diagonal matrix with the given vector; $R = 0.01$.

\subsubsection{Gene regulatory network (GRN)}

MAKO is also applied to a biological gene regulatory network (GRN), which constitutes a synthetic three-gene regulatory network characterized by the oscillatory dynamics of mRNAs and proteins~\cite{sootla2013periodic}. The state vector is $[m_1, m_2, m_3, p_1, p_2, p_3]^\text{T}$, where $m_{1,2,3}$ denotes the concentration of mRNA for the corresponding genes and $p_{1,2,3}$ denotes the concentration of the proteins. The control input $u$ will be executed through light control signals capable of triggering gene expression via the activation of their photosensitive promoters. The controller is expected to maintain the concentration of protein 1 $p_1$ to 6. A more detailed description of the controlled GRN system is referred to \cite{sootla2013periodic}.
In the GRN system, the dissociation constant $K$ and the input scalar $b_1$ to protein 1 are assumed to be uncertain, $K\in [2, 8]$ and  $b_1\in [3, 7]$. Their nominal values are $K= 5$ and $b_1 =5$. The value of $\alpha$ is 1.1, and the weighting matrices are $Q=\text{diag}([0, 0, 0, 1, 0, 0])$, $R = \text{diag}([0.01, 0.01, 0.01])$.

\begin{figure}

        \centering
    \includegraphics[width=0.2\textwidth]{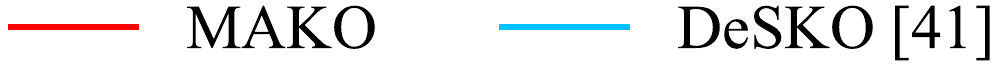}
    \centering
    \subfigure[Cartpole]{
    \includegraphics[width=0.95\columnwidth]{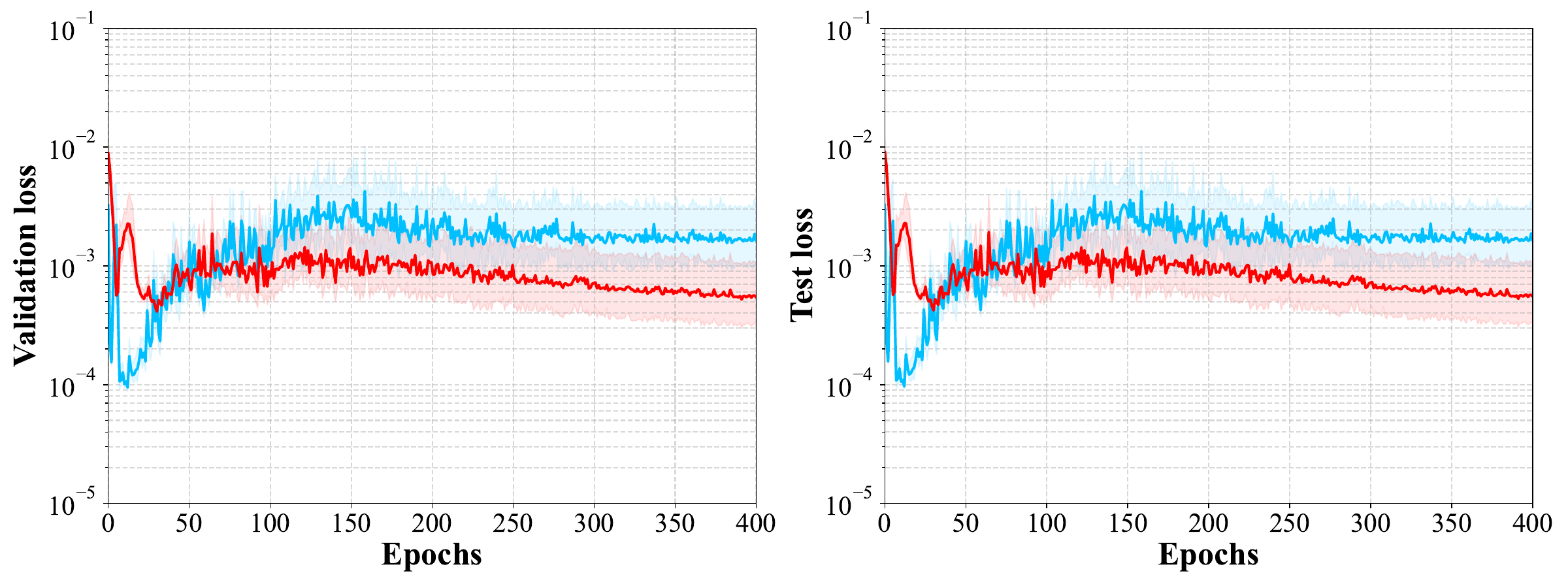}
    }
    \subfigure[GRN]{
    \includegraphics[width=0.95\columnwidth]{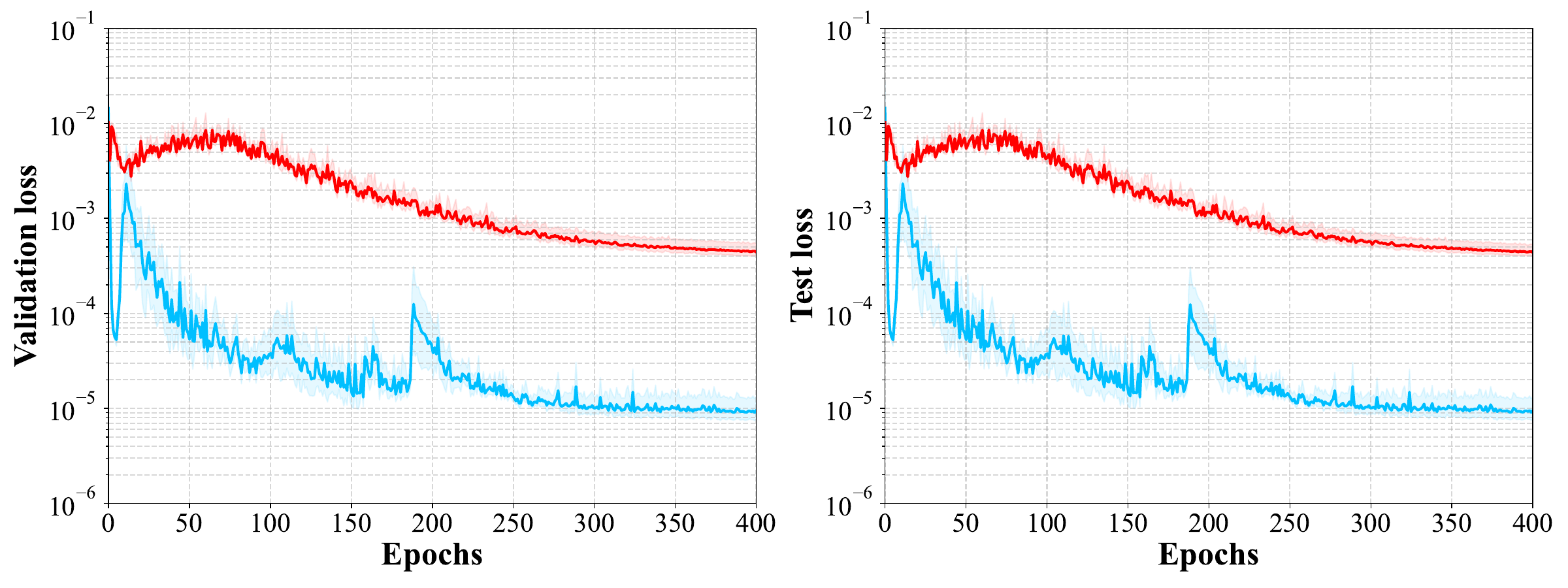}
    }
    \subfigure[Chemical process]{
    \includegraphics[width=0.95\columnwidth]{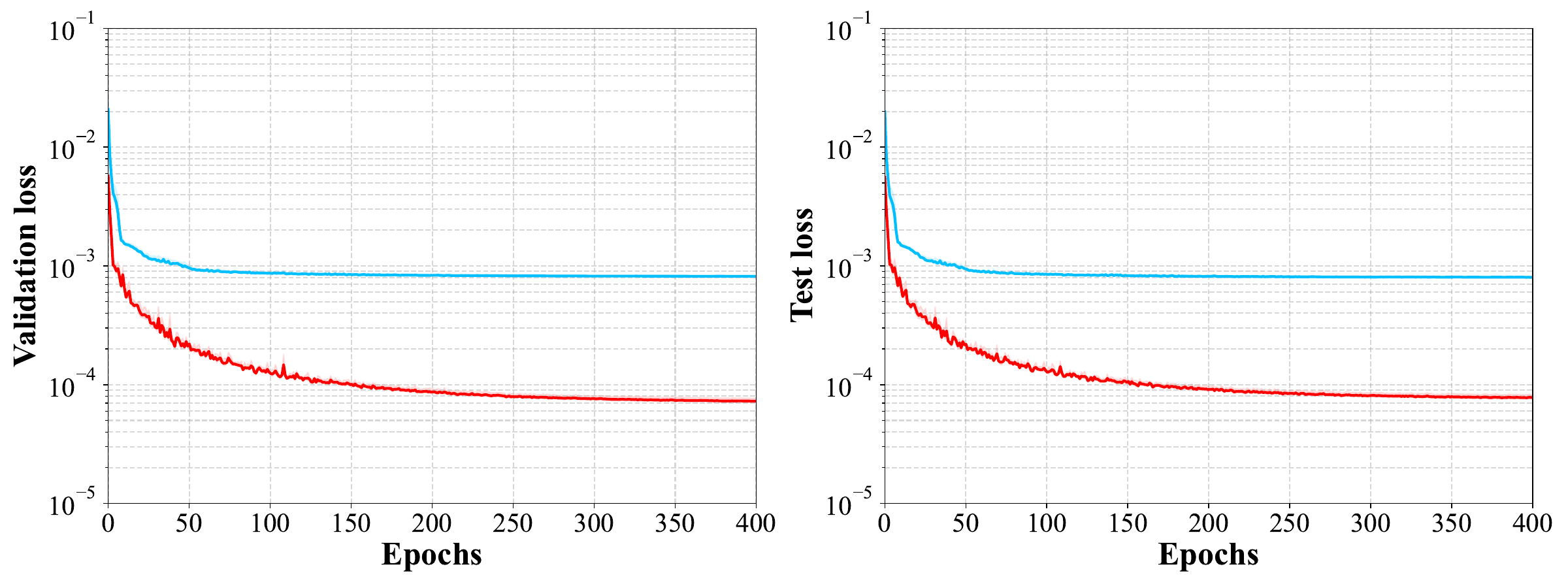}
    }
    \caption{Cumulative prediction errors calculated on both the validation and test datasets. The Y-axis represents the cumulative mean-squared prediction error in a logarithmic scale over 16 time steps, while the X-axis denotes the training epochs. The shaded area illustrates the confidence interval, corresponding to one standard deviation, calculated over 5 random initializations.}
\vspace{-0.2em}
    \label{fig:modeling-validation_loss}
\end{figure}

\subsubsection{Reactor-separator chemical process}

Finally, we apply MAKO to a chemical process that involves two continuously stirred tank reactors and a flash tank separator {\cite{liu2008two}}. A detailed description of this process can be found in {\cite{liu2008two}}. 
The state vector encompasses $[X_{A1}, X_{B1},T_{1},X_{A2},\allowbreak X_{B2},T_{2}, X_{A3}, X_{B3},T_{3}]^\text{T}$, including the mass fractions of $A$ and $B$ denoted by $X_{Ai}$ and $X_{Bi}$, and the temperatures $T_i$, $1,2,3$, across the three vessels. The control objective is to maintain the concentrations of $A$ and $B$ at a steady-state level $x_s= \big[0.18, 0.67, 480.3~\text{K}, 0.19, 0.65, 472.8~\text{K}, 0.06, 0.67, 474.9~\text{K}\big]^\text{T}$. The heating inputs are constrained within $\left[0, 0 , 0\right]^{\text{T}} \times10^6~\text{kJ}/\text{h}\leq u\leq [4.87, 1.68, 4.87]^{\text{T}}\times10^6~\text{kJ}/\text{h}$. Initially, the state is distributed uniformly within the region of $[0.8 x_s, 1.2 x_s]$.
The temperature of the feed stream to reactors 1 and 2 is assumed to be uncertain, $T_{10} \in [150K, 450K] $ and $T_{20}\in [150K, 450K]$. Their nominal values are $T_{10} =300 K $ and $T_{20} =300 K$. The value of $\alpha$ is 1.98, and the weighting matrices are $Q=\text{diag}([1, 1, 0, 1, 1, 0, 1, 1, 0])$, $R = \text{diag}([1, 1, 1])\times10^{-3}$.

{The MAKO model is trained on a multi-modal dataset covering different task settings with randomly sampled inputs. The collected data is normalized by using the mean and standard deviation vectors. In our evaluation, the uncertain parameters are first uniformly sampled from the parameter space, and then the inputs are uniformly sampled from the input space at each time instant to excite the system.}

\subsection{Baseline for comparison}

We compare the performance of MAKO with a competitive baseline, the deep stochastic Koopman operator (DeSKO) method \cite{han2023robust}. DeSKO can provide good modeling and control performance in various systems and has been shown to be robust to system uncertainties. DeSKO models are trained on datasets collected using nominal parameter settings. MAKO models, on the other hand, are trained on a meta-dataset comprising sub-datasets collected using randomly sampled parameter settings. The hyperparameters of MAKO are shown in Table~\ref{table:hyperparameters} in the Appendix.



\subsection{Modeling}

In this part, we first evaluate the modeling performance of MAKO. For both MAKO and DeSKO, 5 models are randomly initialized and trained. Both models are trained for 400 epochs. During each epoch, the models are trained using the Adam optimizer with mini-batches of 128 data points until the entire training dataset for each model has been traversed. The $l_2$ norm of the prediction error on the validation and test datasets in each epoch is presented in Fig.~\ref{fig:modeling-validation_loss}. As observed from Fig.~\ref{fig:modeling-validation_loss}, MAKO demonstrates good modeling performance across different benchmark systems. The average prediction errors of MAKO models over a 16-step horizon are less than $10^{-2}$. The modeling performance of MAKO is consistent on both the validation set and the test set. {While both MAKO and DeSKO provide high modeling accuracy, DeSKO outperforms MAKO in the GRN system. This may be attributed to two factors. 1) The MAKO model is trained and evaluated on a multi-modal dataset containing diverse task settings, while DeSKO is trained specifically under the nominal parameter setting, which allows it to specialize in modeling the GRN dynamics under those specific conditions. 2) MAKO prioritizes adaptability and generalization, which may slightly compromise its predictive accuracy for a specific task.} Moreover, MAKO models surpass DeSKO models in both the Cartpole system and the chemical process. MAKO also exhibits consistent prediction accuracy and low variances across different parameter initializations.

\begin{figure}
        \includegraphics[width=0.32\textwidth]{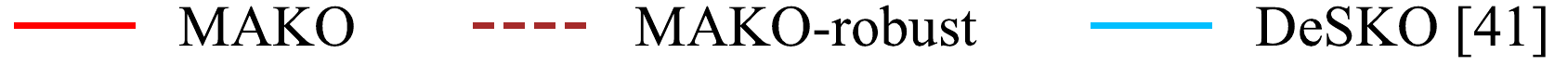}
    \centering
        \includegraphics[width=0.48\textwidth]{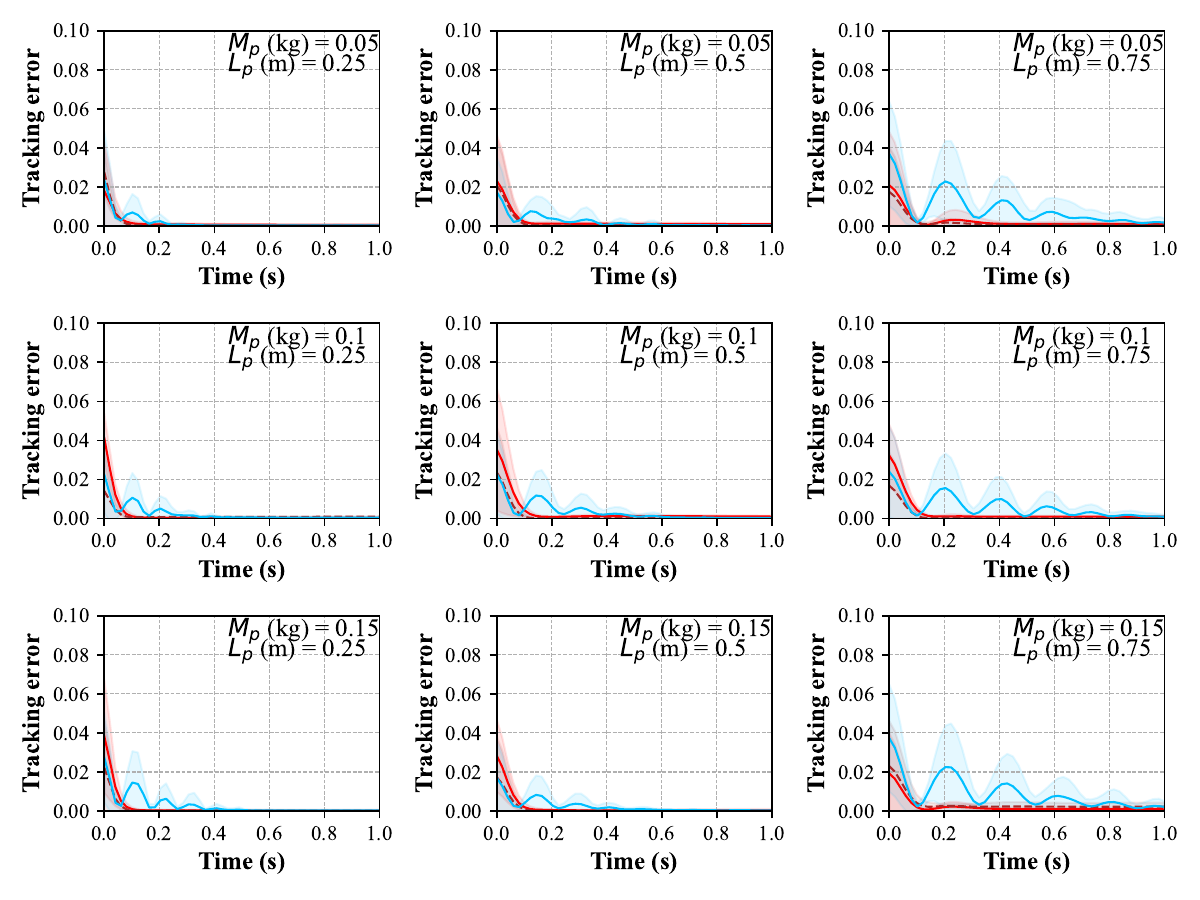}
    \caption{The trajectories of the tracking errors given by MAKO and DeSKO on Cartpole under 9 parameter settings. }
    \label{fig:control_eval_cartpole}
\end{figure}

\subsection{Control}

In this subsection, we examine the performance of the proposed MAKO-based controller. We evaluate both controllers, designed using nominal adaptation and robust adaptation, respectively, and refer to them as MAKO and MAKO-Robust. For each considered system, we uniformly take 9 sets of parameters from the respective parameter space described in Section~\ref{subsec:Simulation setup and baselines}. These parameter settings were not encountered by MAKO during its training phase, which can showcase the generalization ability of MAKO. On the other hand, the nominal parameter setting, which was encountered during DeSKO's training, is also included for comparison purposes. The cumulative tracking errors and the trajectories of the norm of tracking errors of MAKO and DeSKO for the three considered systems are shown in Fig.~\ref{fig:control_eval_cartpole}, Fig.~\ref{fig:control_eval_grn}, and Fig.~\ref{fig:control_eval_cstr}, respectively.
Fig.~\ref{fig:control_eval_cartpole} to Fig.~\ref{fig:control_eval_cstr} demonstrate that MAKO achieves good control performance in all three benchmark systems under various parameter settings. Its cumulative tracking error is also lower than or comparable to that of DeSKO. In the Cartpole example, DeSKO achieves stabilizing performance with higher cumulative costs compared to MAKO. For the second system, DeSKO accurately tracks the given reference in GRN under 3 sets of parameter settings. In the chemical process example, DeSKO fails to stabilize the tracking error in all parameter settings. Compared to the nominal MAKO-based MPC, the MPC leveraging the robust adaptation scheme exhibits faster and more stable transient behavior, while achieving comparable or smaller steady-state tracking errors.
In the current work, the simulation is conducted on a computer equipped with an i7-12700 2.10 GHz CPU. For the Cartpole system, the MAKO-robust framework achieved an average computation time of 0.0203 seconds per time step, which demonstrates its suitability for real-time control applications.

\begin{figure}
\includegraphics[width=0.32\textwidth]{legend3.pdf}
    \centering
    \includegraphics[width=0.48\textwidth]{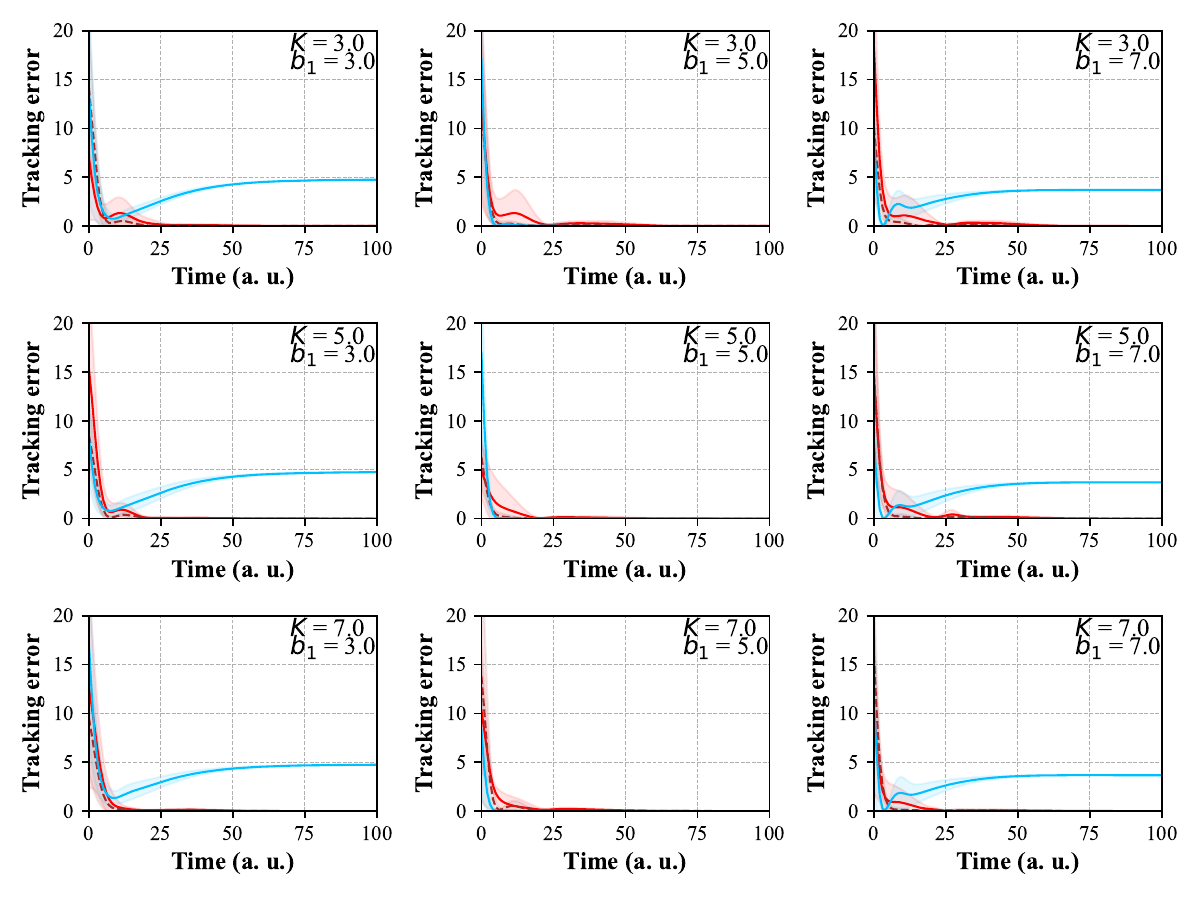}
    \caption{The trajectories of the tracking errors given by MAKO and DeSKO on GRN under 9 parameter settings. 
    }
    \label{fig:control_eval_grn}
\end{figure}

\begin{remark}
    While Assumption~\ref{assumption: 1} can be difficult to verify on the benchmark systems, MAKO provides good performance in all three case studies, demonstrated by high modeling accuracy and robust control performance despite the lack of strict forward invariance. The simulation results demonstrate the framework's ability to address real-world complexities and suggest its potential to handle a broader range of systems where theoretical guarantees on invariance may not strictly apply.
\end{remark}

\section{Conclusion}

In this paper, meta-learning and the Koopman operator were integrated for the first time to develop a multi-modal modeling approach for parametrically uncertain nonlinear systems.  An adaptation scheme was designed to refine the meta-trained Koopman operator with online data, and the convergence of parameter approximation and state prediction errors was proven under mild assumptions. Based on the proposed MAKO model, an AMPC scheme was proposed; this control scheme ensures the stability of the closed-loop system in the presence of previously seen and unseen parameter settings. Through extensive simulation evaluations, we demonstrated that the proposed MAKO modeling and control framework can outperform the baseline methods.

\begin{figure}
\includegraphics[width=0.32\textwidth]{legend3.pdf}
    \centering
    \includegraphics[width=0.48\textwidth]{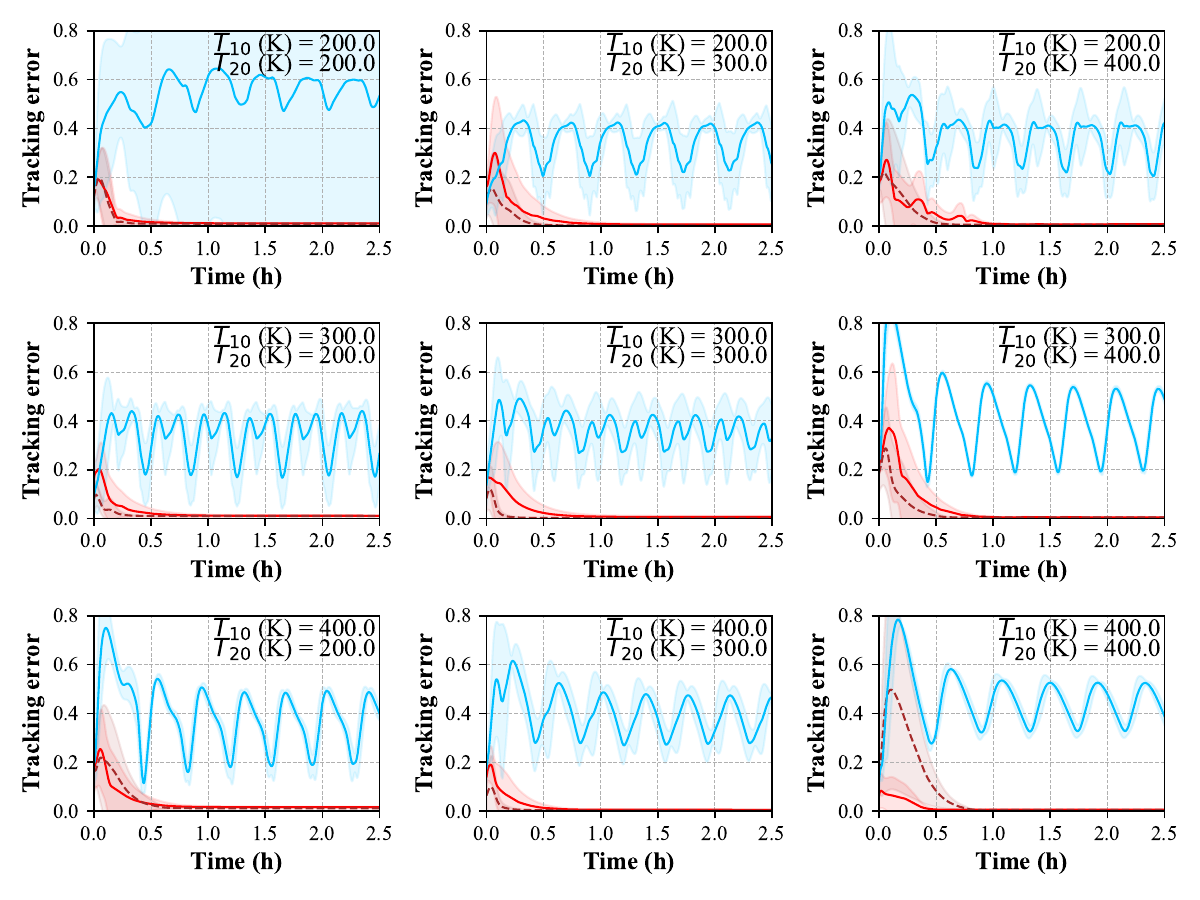}
    \caption{The trajectories of the tracking errors given by MAKO and DeSKO on the chemical process under 9 parameter settings. }
    \label{fig:control_eval_cstr}
\end{figure}

{We identify the following potential topics for future research: 1) While MAKO demonstrated good performance in simulations, applying this method to real-world systems with parametric uncertainties would be of interest in future research.  2) A formal analysis of persistent excitation (PE) requirements could be crucial for examining training convergence. Specifically, understanding how trajectory richness and PE influence the quality of the learned Koopman operator and its ability to generalize across unseen tasks would provide valuable theoretical insights. 3) A more systematic investigation of trajectory
length in relation to the meta-learning of Koopman operators is an important direction of future research. 4) Extending the proposed approach to higher-dimensional systems would be another interesting topic for future exploration.}



\begin{table}[htb]
\caption{Hyperparameters of MAKO}\label{table:hyperparameters}
\begin{tabular}{p{3.3cm}|c|c|c}
\hline
Hyperparameters&Cartpole& GRN & \makecell{Chemical \\ process} \\
\hline
Observable dimension & 128&128&256\\
\hline
Trajectory length & 250&400&500\\
\hline
Noise upper bound $\epsilon_w$ & $10^{-6}$&$10^{-6}$&$10^{-8}$\\
\hline
Noise upper bound $\epsilon_v$ & $10^{-4}$&$10^{-4}$&$10^{-8}$\\
\hline
Number of sub-datasets $N$ & $10$ &$10$&$20$\\
\hline
Size of sub-dataset $\mathcal{D}_i$&\multicolumn{3}{c}{$5 \times 10^4$} \\
\hline
Batch Size & \multicolumn{3}{c}{128}\\
\hline
Learning rate & \multicolumn{3}{c}{$10^{-4}$}\\
\hline
Prediction horizon $H$ & \multicolumn{3}{c}{16} \\
\hline
Structure of $\psi_\theta(\cdot)$ &\multicolumn{3}{c}{(128,128)}\\
\hline
Activation function & \multicolumn{3}{c}{ReLU}\\
\hline
$l_2$ norm regularization coefficient & \multicolumn{3}{c}{0.001}\\
\hline
\end{tabular}
\end{table}


\end{document}